\documentclass[12pt]{article}
\usepackage{enumerate}
\usepackage{amsfonts}
\usepackage{amsmath}
\usepackage{amssymb}
\usepackage{amsthm}

\def\H{\mathcal{H}}
\def\K{\mathcal{K}}

\def\S{\mathfrak{S}}
\def\C{\mathfrak{C}}
\def\T{\mathfrak{T}}
\def\B{\mathfrak{B}}

\newcommand{\supp}{\mathrm{supp}}
\newcommand{\rank}{\mathrm{rank}}

\newcommand{\Tr}{\mathrm{Tr}}

\newcommand{\ri}{\mathrm{ri}}

\newcounter{defin}  \newcounter{lemma}  \newcounter{theorem}
\newcounter{property} \newcounter{corol}  \newcounter{remark} \newcounter{example}

\newenvironment{lemma}{\par\refstepcounter{lemma}
     \textbf{Lemma \thelemma.} }{\rm\par}
\newenvironment{theorem}{\par\refstepcounter{theorem}
     \textbf{Theorem \thetheorem.}\ }{\rm\par}
\newenvironment{property}{\par\refstepcounter{property}
     \textbf{Proposition \theproperty.}\ }{\rm\par}
\newenvironment{corollary}{\par\refstepcounter{corol}
     \textbf{Corollary \thecorol.} }{\rm\par}
\newenvironment{definition}{\par\refstepcounter{defin}
     \textbf{Definition \thedefin.}\ }{\rm\par}
\newenvironment{remark}{\par\refstepcounter{remark}
     \textbf{Remark \theremark.}}{\rm\par}
\newenvironment{example}{\par\refstepcounter{example}
     \textbf{Example \theexample.}}{\rm\par}

\begin{document}
\title{On superactivation of zero-error capacities and reversibility of a quantum channel}
\author{M.E. Shirokov\footnote{Steklov Mathematical Institute, RAS, Moscow, email:msh@mi.ras.ru}, T.V.
Shulman\footnote{Institute of Mathematics, Polish Academy of
Sciences, Sniadeckich 8, 00-956 Warszawa, Poland, email:
tshulman@impan.pl}}
\date{}
\maketitle

\begin{abstract}
We propose examples of low dimensional quantum channels
demonstrating different forms of superactivation of one-shot
zero-error capacities, in particular, the extreme superactivation
(this complements the recent result of T.S.Cubitt and G.Smith).

We also describe classes of quantum channels whose zero-error
classical and quantum capacities cannot be superactivated.

We consider implications of the superactivation of one-shot
zero-error capacities to analysis of reversibility of a
tensor-product channel with respect to families of pure states.

Our approach based on the notions of complementary channel and of
transitive subspace of operators makes it possible to study the
superactivation effects for infinite-dimensional channels as well.
\end{abstract}
\maketitle


\section{Introduction}

The effect of superactivation of quantum channel capacities is one
of the main recent discoveries in quantum information theory. It
means that the particular capacity of tensor product of two quantum
channels may be positive despite the same capacity of each of these
channels is zero.

This effect was originally observed by G.Smith and J.Yard in
\cite{S&Y}, who gave examples of two channels $\Phi$ and $\Psi$ with
zero quantum capacity such that the channel $\Phi\otimes\Psi$ has
positive quantum capacity.

The same phenomenon for the (one shot and asymptotic) zero-error
classical capacities was established by T.S.Cubitt, J.Chen and W.A.
Harrow in \cite{CCH}. Simultaneously and independently R.Duan
presented an example of low dimensional channels demonstrating
superactivation of the one-shot zero-error classical capacity
\cite{Duan}.

The extreme form of superactivation of zero-error capacities was
observed by T.S.Cubitt and G.Smith in \cite{C&S}, who proved the
existence of two channels $\Phi$ and $\Psi$ with zero (asymptotic)
zero-error classical capacity such that the channel
$\Phi\otimes\Psi$ has positive zero-error quantum capacity.

In this paper we present examples of low dimensional quantum
channels which demonstrate different forms of superactivation of
one-shot zero-error capacities. In particular, in Corollary
\ref{main-c+} we give a \emph{symmetric} example of superactivation
of one-shot zero-error classical capacity with the minimal possible
input dimension $\dim\H_{A} = 4$ and the minimal Choi rank
$\dim\H_{E}=3$ so that $\dim\H_{B}\leq12$ (this answers the question
stated after Theorem 1 in \cite{Duan}). As to the extreme form of
superactivation of one-shot zero-error capacities, the existence of
such channels in high dimensions follows from the results in
\cite{C&S}. However nothing was known about their minimal
dimensions. Here (Corollary \ref{main-c}) we give an explicit
example with $\dim\H_{A}=8$, $\dim\H_{E}=5$ and $\dim\H_{B}\leq40$.

The aim of this paper is also to point out the relation between the
superactivation of one-shot zero-error capacities and results on
transitive and reflexive subspaces of operators \cite{DMR,M&S}. In
fact, the notion of transitive subspace is very close to the notion
of unextendible subspace traditionally used in analysis of the
superactivation (one can easily show that in finite dimensions they
are related by the natural isomorphism between the tensor product
$\H\otimes\K$ of two Hilbert spaces and the space of all operators
from $\H$ to $\K$). Nevertheless, the recent results concerning
transitive subspaces of operators (presented in \cite{DMR}) seem to
be unknown for specialists in quantum information theory. It is also
essential that these results can be used for analysis of
superactivation effects for infinite dimensional quantum channels.

Some results concerning transitive and reflexive subspaces of
operators can also be applied for showing that channels of certain
type cannot be superactivated by any other channels. A result in
this direction was obtained recently by J.Park and S.Lee in
\cite{P&L}. They showed that superactivation of one-shot zero-error
classical capacity is not possible if one of two channels is a qubit
channel. Our approach gives a very simple proof of this result and
also allows us to prove similar statements for some other important
classes of channels (Proposition \ref{non-sa}, Corollary
\ref{non-sa-c-2}). We also describe classes of channels for which
the superactivation of one-shot and asymptotic zero-error quantum
capacities does not hold (Proposition \ref{non-sa+}, Corollary
\ref{non-sa+c-2}).

In this paper we also  consider the relations between positivity of
one-shot  classical and quantum zero-error capacities of a quantum
channel and reversibility properties of this channel with respect to
families of pure states. These relations show that the
superactivation of one-shot  classical (correspondingly, quantum)
zero-error capacities is equivalent to "superactivation" of
reversibility of a channel with respect to orthogonal
(correspondingly, non-orthogonal) families of pure states. It is
observed that such superactivation of reversibility with respect to
\emph{complete} families of pure states is not possible (Proposition
\ref{eq-1}).

\section{On positivity of classical and quantum zero-error capacities of a quantum channel}

Let $\H$ be a separable\footnote{In the main part of the paper we
may assume that these spaces are finite-dimensional, although all
the results are valid in infinite dimensions if we accept the value
$"+\infty"$ for $\bar{C}_0(\Phi)$, $\bar{Q}_0(\Phi)$, etc. The case
of infinite-dimensional quantum channels is included because of our
intension to study reversibility properties of a tensor product
channel (Section 5).} Hilbert space, $\B(\H)$ and $\mathfrak{T}(
\mathcal{H})$ -- the Banach spaces of all bounded operators in
$\mathcal{H}$ and of all trace-class operators in $\H$
correspondingly, $\S(\H)$ -- the closed convex subset of
$\mathfrak{T}( \H)$ consisting of positive operators with unit trace
called \emph{states} \cite{H-SCI,N&Ch}. If $\dim\H=n<+\infty$ we may
identify $\B(\H)$ and $\T(\H)$ with the space $\mathfrak{M}_n$ of
all $n\times n$ matrices (equipped with the appropriate norm).
\smallskip

Let
$\Phi:\mathfrak{T}(\mathcal{H}_A)\rightarrow\mathfrak{T}(\mathcal{H}_B)$
be a quantum channel, i.e. a  completely positive trace-preserving
linear map \cite{H-SCI,N&Ch}. The \textit{dual} channel
$\Phi^{*}:\mathfrak{B}(\H_{B})\rightarrow\mathfrak{B}(\H_{A})$
(defined by the relation $\Tr\Phi(\rho)B=\Tr \rho\,\Phi^{*}(B)$,
$\rho\in\T(\H_{A})$, $B\in\mathfrak{B}(\H_{B})$) is a completely
positive map such that $\Phi^{*}(I_{\H_B})=I_{\H_A}$.

The Stinespring theorem implies the existence of a Hilbert space
$\mathcal{H}_E$ and of an isometry
$V:\mathcal{H}_A\rightarrow\mathcal{H}_B\otimes\mathcal{H}_E$ such
that
\begin{equation}\label{Stinespring-rep}
\Phi(\rho)=\mathrm{Tr}_{\mathcal{H}_E}V\rho V^{*},\quad
\rho\in\mathfrak{T}(\mathcal{H}_A).
\end{equation}
The quantum  channel
\begin{equation}\label{c-channel}
\mathfrak{T}(\mathcal{H}_A)\ni
\rho\mapsto\widehat{\Phi}(\rho)=\mathrm{Tr}_{\mathcal{H}_B}V\rho
V^{*}\in\mathfrak{T}(\mathcal{H}_E)
\end{equation}
is called \emph{complementary} to the channel $\Phi$
\cite{H-SCI,H-c-c}. The complementary channel is defined uniquely up
to isometrical equivalence \cite[the Appendix]{H-c-c}.

The Stinespring representation (\ref{Stinespring-rep}) generates the
Kraus representation
\begin{equation}\label{Kraus-rep}
\Phi(\rho)=\sum_{k=1}^{\dim\H_E}V_{k}\rho V^{*}_{k},\quad
\rho\in\,\T(\H_A),
\end{equation}
in which  $\{V_{k}\}$ is a set of linear operators from $\H_A$ into
$\H_B$ such that $\sum_{k}V^{*}_{k}V_{k}=I_{\H_A}$. The operators
$V_k$ are defined by the relation
$$
\langle\varphi|V_k\psi\rangle=\langle\varphi\otimes
k|V\psi\rangle,\quad\varphi\in\H_B,\psi\in\H_A,
$$
where  $\{|k\rangle\}$ is an orthonormal basis in the space $\H_E$.
The complementary channel (\ref{c-channel}) can be expressed via
these operators as follows
\begin{equation}\label{Kraus-rep-c}
\widehat{\Phi}(\rho)=\sum_{k,l=1}^{\dim\H_E}\Tr\left[V_{k}AV_{l}^{*}\right]|k\rangle\langle
l|,\quad \rho\in\,\T(\H_A).
\end{equation}
Among different Stinespring representations (\ref{Stinespring-rep})
of a given channel $\Phi$ there are representations with the
environment space $\H_E$ of minimal dimension (such representations
are called minimal \cite{H-c-c}). They generates Kraus
representations (\ref{Kraus-rep}) with the minimal number of nonzero
summands called \emph{Choi rank} of the channel $\Phi$
\cite{H-SCI,N&Ch}. We assume in what follows that
(\ref{Stinespring-rep}) is a minimal Stinespring representation, so
that $\dim\H_E$ is the Choi rank of $\Phi$.\smallskip

The one-shot zero-error classical capacity $\bar{C}_0(\Phi)$ of a
channel $\Phi$ can be defined as $\;\sup_{\S\in
c_0(\Phi)}\log\sharp(\S)\,$, where $c_0(\Phi)$ is the set of all
families $\{\rho_i\}$ of input states such that
$\,\supp\Phi(\rho_i)\perp\supp\Phi(\rho_j)\,$ for all $i\neq
j$.\footnote{The support $\mathrm{supp}\rho$ of a state $\rho$ is
the orthogonal complement to its kernel.} The (asymptotic)
zero-error classical capacity is defined by regularization:
$C_0(\Phi)=\sup_n n^{-1}\bar{C}_0(\Phi^{\otimes n})$ \cite{CCH, C&S,
Duan, W&Co, ZEC, P&L}. \smallskip

Let $\varphi,\psi\in\H_A$. It follows from (\ref{Stinespring-rep}), (\ref{c-channel})
and the Schmidt decomposition of the vectors $V\varphi$ and $V\psi$ in
$\mathcal{H}_B\otimes\mathcal{H}_E$ that
\begin{equation}\label{eq-1}
\supp\Phi(|\varphi\rangle\langle
\varphi|)\perp\supp\Phi(|\psi\rangle\langle\psi|)\quad\Leftrightarrow\quad\widehat{\Phi}(|\varphi\rangle\langle
\psi|)=0.
\end{equation}

This observation implies the following lemma.
 \smallskip
\begin{lemma}\label{rank1}
\emph{A channel $\,\Phi:\T(\H_A)\rightarrow\T(\H_B)$ has positive
one-shot zero-error classical capacity if and only if
$\,\ker\widehat{\Phi}$ contains a 1-rank operator.}
\end{lemma}\smallskip

The assertion of Lemma \ref{rank1} agrees with Lemma 1 in
\cite{Duan}, since representation (\ref{Kraus-rep-c}) shows that the
subspace $\,\widehat{\Phi}^*(\B(\H_E))$ is precisely the
\emph{noncommutative graph} $\mathcal{G}(\Phi)$ of the channel
$\Phi$ which is defined as the subspace of $\B(\H_A)$ spanned by the
family of operators $\{V^*_j V_k\}_{kj}$, where $\{V_k\}_{k}$ is a
family of operators from the Kraus representation (\ref{Kraus-rep})
of the channel $\Phi$ \cite[Lemma 1]{W&Co}.\smallskip

\begin{definition}\label{trans-d}\cite{DMR} A subspace
$\mathfrak{L}\subseteq\B(\H)$ is (topologically) \emph{transitive}
if for any vector $\varphi\in\H$ the set
$\mathfrak{L}|\varphi\rangle\doteq\{A|\varphi\rangle\,|\,A\in\mathfrak{L}\}\,$
is dense in $\H$.
\end{definition}\smallskip

If $\,\H$ is a finite-dimensional space then "is dense in" in the
above definition may be replaced by "coincides with".
\smallskip

The following lemma is our basic tool for studying the one-shot zero-error classical capacity.

\begin{lemma}\label{trans-l}
\emph{A channel $\Phi:\T(\H_A)\rightarrow\T(\H_B)$ has positive
one-shot zero-error classical capacity if and only if the
noncommutative graph
$\,\mathcal{G}(\Phi)\doteq\widehat{\Phi}^*(\B(\H_E))$ is not
transitive.}
\end{lemma}

\begin{proof} It is easy to check that a subspace $\mathfrak{L}$ of $\B(\H)$  is
transitive if and only if the subspace
$\mathfrak{L}^{\perp}=\{A\in\T(\H)\,|\,\Tr AB=0\;\forall
B\in\mathfrak{L}\}$ does not contain any 1-rank operator (this was
first noticed by Azoff \cite{Azoff}, see also \cite{DMR}). Now the
statement follows from Lemma \ref{rank1}.
\end{proof}

The one-shot zero-error quantum capacity $\bar{Q}_0(\Phi)$ of a
channel $\Phi$ can be defined as $\;\sup_{\H\in
q_0(\Phi)}\log\dim\H\,$, where $q_0(\Phi)$ is the set of all
subspaces $\H_0$ of $\H_A$ on which the channel $\Phi$ is perfectly
reversible (in the sense that there is a channel $\Psi$ such that
$\Psi(\Phi(\rho))=\rho$ for all states $\rho$ supported by $\H_0$,
see \cite[Ch.10]{H-SCI}). The (asymptotic) zero-error quantum
capacity is defined by regularization: $Q_0(\Phi)=\sup_n
n^{-1}\bar{Q}_0(\Phi^{\otimes n})$ \cite{CCH, C&S, Duan, W&Co, ZEC,
P&L}.

Hence the one-shot zero-error quantum capacity $\bar{Q}_0(\Phi)$ is
positive if and only if there exists a nontrivial subspace $\H_0$ of
$\H_A$ such that the restriction of the channel $\widehat{\Phi}$ to
the subset $\S(\H_0)$ is completely depolarizing
\cite[Ch.10]{H-SCI}, i.e.
$\widehat{\Phi}(\rho_1)=\widehat{\Phi}(\rho_2)$ for all states
$\rho_1$ and $\rho_2$ supported by $\H_0$.\smallskip

These arguments imply the following modification of Lemma 1 in
\cite{C&S}.\smallskip

\begin{lemma}\label{trans-l+}
\emph{A channel $\Phi:\T(\H_A)\rightarrow\T(\H_B)$ has positive
one-shot zero-error quantum  capacity if and only if there are unit
vectors $\varphi$ and $\psi$ in $\H_A$ such that
\begin{equation}\label{operators}
\widehat{\Phi}(|\varphi\rangle\langle\psi|)=0\quad\textit{and}\quad
\widehat{\Phi}(|\varphi\rangle\langle\varphi|)=\widehat{\Phi}(|\psi\rangle\langle\psi|)
\end{equation}
or, equivalently,
\begin{equation}\label{operators+}
\langle \psi|A|\varphi\rangle=0\quad\textit{and}\quad \langle
\varphi|A|\varphi\rangle=\langle \psi|A|\psi\rangle\quad\forall
A\in\mathcal{G}(\Phi)\doteq\widehat{\Phi}^*(\B(\H_E)).
\end{equation}}
\end{lemma}
\begin{proof} It is easy to see that
$\,\widehat{\Phi}(|\varphi\rangle\langle\psi|)=0\,$ if and only if
$$
\widehat{\Phi}(\rho)=\langle\varphi|\rho|\varphi\rangle\widehat{\Phi}(|\varphi\rangle\langle\varphi|)+\langle\psi|\rho|\psi\rangle
\widehat{\Phi}(|\psi\rangle\langle\psi|)
$$
for all states $\rho$ supported by the subspace $\H_{\varphi,\psi}$
spanned by the vectors $\varphi$ and $\psi$. Hence (\ref{operators})
holds if and only if the restriction of the channel $\widehat{\Phi}$
to the subset $\S(\H_{\varphi,\psi})$ is completely depolarizing.
\end{proof}

Lemmas \ref{trans-l} and \ref{trans-l+} imply the following
conditions for positivity of the one-shot  classical and quantum
zero-error capacities.\smallskip

\begin{property}\label{com}
\emph{Let $\,\Phi:\T(\H_A)\rightarrow\T(\H_B)$ be a quantum channel
and $\,\mathcal{G}(\Phi)\doteq\widehat{\Phi}^*(\B(\H_E))$ its
noncommutative graph. Then}
\begin{equation}\label{imp-1}
[\mathcal{G}(\Phi)]'\;\,\textit{is non-trivial}\; (\neq\{\lambda
I\}) \quad\Rightarrow\quad \bar{C}_0(\Phi)>0,
\end{equation}
\begin{equation}\label{imp-2}
[\mathcal{G}(\Phi)]' \;\,\textit{is
noncommutative}\quad\quad\;\;\,\Rightarrow \quad \bar{Q}_0(\Phi)>0.
\end{equation}

\emph{If $\,\mathcal{G}(\Phi)$ is an algebra then
$"\Leftrightarrow"$ holds in the above implications.}
\end{property}\medskip

\begin{remark}\label{com-r}
In general $"\Leftrightarrow"$ does not hold in (\ref{imp-1}) and
(\ref{imp-2}). There exists a channel $\Phi$ with
$\bar{Q}_0(\Phi)>0$ for which $[\mathcal{G}(\Phi)]'=\{\lambda I\}$.
Indeed, since the subspace of $\mathfrak{M}_4$ consisting of the
matrices
$$
\left[\begin{array}{cc}\lambda I_2& A\\B&C\end{array}\right],\quad
A,B,C\in \mathfrak{M}_2,
$$
is symmetric and contains the unit matrix $I_4$, Proposition
\ref{cmp} below (or Lemma 2 in \cite{Duan}) shows that this subspace
is the noncommutative graph of some channel $\Phi$. It follows from
Lemma \ref{trans-l+} that $\bar{Q}_0(\Phi)>0$, but it is easy to see
that the commutant of this subspace is trivial.
\end{remark}

\begin{proof} If the algebra $[\mathcal{G}(\Phi)]'$ is non-trivial, then it contains a non-trivial
projection $P$. Then $\mathcal{G}(\Phi) P(\H_A) \subseteq P(\H_A)$
and hence $\mathcal{G}(\Phi)$ is not transitive.
The first implication follows now from Lemma
\ref{trans-l}.

If the algebra
$[\mathcal{G}(\Phi)]'$ is noncommutative, then,  by Lemma \ref{e-p} below,  there exists
a partial isometry $W\in[\mathcal{G}(\Phi)]'$ such that the
projections $P=W^*W$ and $Q=WW^*$ are orthogonal. Let
$|\varphi\rangle$ be an arbitrary vector in $P(\H_A)$ and
$|\psi\rangle=W|\varphi\rangle\in Q(\H_A)$. Then it is easy to see
that (\ref{operators+}) holds and
by Lemma \ref{trans-l+} the second implication follows.

By Lemma \ref{trans-l}  $\bar{C}_0(\Phi)>0$ implies the existence of
a non-zero vector $\varphi$ such that $\H_{\varphi}=
\overline{\{A|\varphi\rangle,\, A\in\mathcal{G}(\Phi)\}}\neq\H_A$.
If $\,\mathcal{G}(\Phi)$ is an algebra then $\H_{\varphi}$ is an
invariant subspace   for $\mathcal{G}(\Phi)$. Since the algebra
$\mathcal{G}(\Phi)$ is symmetric, it implies that the orthogonal
projection onto $\H_{\varphi}$ commutes with $\mathcal{G}(\Phi)$.
\smallskip

Suppose $\,\mathcal{G}(\Phi)$ is an algebra and
$\,\bar{Q}_0(\Phi)>0\,$. We will show that $[\mathcal{G}(\Phi)]'$
contains two orthogonal equivalent projections and hence is
noncommutative. By Lemma \ref{trans-l+} there are vectors  $\varphi$
and $\psi$ in $\H_A$ such that (\ref{operators+}) holds. Let
$\H_{\varphi} = \overline{\{A|\varphi\rangle\; | \; A \in
\mathcal{G}(\Phi)\}}$ and $\H_{\psi} = \overline{\{A|\psi\rangle\; |
\; A \in \mathcal{G}(\Phi)\}}$. It follows from (\ref{operators+})
that $\H_{\varphi}\perp\H_{\psi}$ and that
$\|A|\varphi\rangle\|=\|A|\psi\rangle\|$ for all
$A\in\mathcal{G}(\Phi)$. Hence the operator  $W$ defined by the
relations
$$
WA|\varphi\rangle=A|\psi\rangle\quad\forall
A\in\mathcal{G}(\Phi)\quad\textup{and}\quad
W|\phi\rangle=0\quad\forall \phi\in\H^{\perp}_{\varphi}
$$
is a partial isometry for which $\H_{\varphi}$ and $\H_{\psi}$ are
initial and final subspaces. Since these subspaces are invariant for
all operators in $\mathcal{G}(\Phi)$, it is easy to see that
$W\in[\mathcal{G}(\Phi)]'$. Thus, the algebra $[\mathcal{G}(\Phi)]'$
contains the orthogonal equivalent projections $W^*W$ and $WW^*$
(onto $\H_{\varphi}$ and $\H_{\psi}$ respectively).
\end{proof}

\begin{lemma}\label{e-p}\footnote{We
are grateful to V.S.Shulman for this observation.} \emph{A von
Neumann algebra $\mathfrak{M}$ is noncommutative if and only if it
contains two orthogonal equivalent projections.}\footnote{Two
projections $P$ and $Q$ are said to be equivalent relative to a von
Neumann algebra $\mathfrak{M}$ when $P=W^*W$ and $Q=WW^*$ for some
$W\in\mathfrak{M}$ \cite[Definition 6.1.4]{K&R}.}
\end{lemma}

\begin{proof} If $\mathfrak{M}$ is noncommutative then it contains a
noncentral projection $P$. Let $\bar{P}=I-P$. By the Comparison
Theorem \cite[Theorem 6.2.7.]{K&R} there exists a central projection
$E$ such that $PE\preceq\bar{P}E$ and $\bar{P}\bar{E}\preceq
P\bar{E}$, where $\bar{E}=I-E$ and $"\preceq"$ denotes the
projection ordering (relative to $\mathfrak{M}$) \cite{K&R}. Since
$P$ is noncentral, either $PE\neq0$ or $\bar{P}\bar{E}\neq0$
(otherwise $P=\bar{E}$).

If $PE\neq0$ then $PE$ is equivalent to some projection
$Q\leq\bar{P}E$. It is clear that the projections $PE$ and $Q$ are
orthogonal.

If $\bar{P}\bar{E}\neq0$ then the similar arguments shows the
existence of a projection $Q'\leq P\bar{E}$ equivalent to
$\bar{P}\bar{E}$.
\end{proof}

\begin{example}\label{gauss} An important class of channels
for which $"\Leftrightarrow"$ hods in (\ref{imp-1}) and in
(\ref{imp-2}) consists of Bosonic Gaussian channels defined as
follows.

Let $\mathcal{H}_{X}$ $(X=A,B)$ be the space of irreducible
representation of the Canonical Commutation Relations (CCR)
\begin{equation*}
W_X(z)W_X(z^{\prime })=\exp
\left(-\textstyle{\frac{\mathrm{i}}{2}}\,\Delta_{X}(z,z^{\prime})\right)
W_X(z^{\prime }+z),\quad z,z'\in Z_X,
\end{equation*}
where  $(Z_{X},\Delta _{X})$ is a symplectic space and $W_{X}(z)$
are the Weyl operators \cite{Caruso,E&W},\cite[Ch.12]{H-SCI}. Denote
by $s_X$ the number of modes of the system $X$, i.e. $2s_X=\dim
Z_X$. A Bosonic Gaussian channel  $\Phi_{K,l,\alpha}
:\mathfrak{T}(\mathcal{H}_{A})\rightarrow
\mathfrak{T}(\mathcal{H}_{B})$ is defined via the action of its dual
$\Phi_{K,l,\alpha}^{\ast }:\mathfrak{B}(\mathcal{H}_{B})\rightarrow
\mathfrak{B}(\mathcal{H}_{A})$ on the Weyl operators:
$$
\Phi_{K,l,\alpha}^{\ast}(W_{B}(z))=W_A(Kz)\exp \left[\,
\mathrm{i}\hspace{1pt}l\hspace{1pt}z-\textstyle\frac{1}{2}\hspace{1pt}z^{\top
}\alpha \hspace{1pt}z\,\right],\quad z\in Z_B,
$$
where $K:Z_{B}\rightarrow Z_{A}$ is a linear operator, $l\,$ is a
$\,2s_B$-dimensional real row and $\,\alpha\,$ is a real symmetric
$\,(2s_B)\times(2s_B)$ matrix satisfying the inequality $\alpha \geq
\pm \frac{\mathrm{i}}{2}\left[ \Delta _{B}-K^{\top }\Delta
_{A}K\right]$ \cite{Caruso,E&W,H-SCI}.\smallskip

Any Bosonic Gaussian channel $\Phi_{K,l,\alpha}$ is unitary
equivalent to the channel $\Phi_{K,0,\alpha}$ for which Bosonic
unitary dilation always exists \cite{Caruso,H-SCI}. So, Lemma 2 in
\cite{BRC} shows that the noncommutative graph of the channel
$\Phi_{K,0,\alpha}$ coincides with the algebra generated by the
family $\{W_A(z)\}_{z\in [K(\ker\alpha)]^{\mathrm{c}}}$ of Weyl
operators in $\H_A$, where $[K(\ker\alpha)]^{\mathrm{c}}$ is the
skew-orthogonal complement to the subspace $K(\ker\alpha)\subseteq
Z_A$. It follows that
$[\mathcal{G}(\Phi_{K,0,\alpha})]'=\left[\{W_A(z)\}_{z\in
K(\ker\alpha)}\right]''$.\smallskip

Since $\ker K\cap\ker\alpha=\{0\}$ and
$\Delta_A(Kz_1,Kz_2)=\Delta_B(z_1,z_2)$ for all $z_1,z_2 $ in
$\ker\alpha$ (see \cite[Ch.12]{H-SCI} or \cite[Lemma 2]{BRC}), the
algebra $\left[\{W_A(z)\}_{z\in K(\ker\alpha)}\right]''$ is
nontrivial if and only if $\ker\alpha\neq \{0\}$ and it is
noncommutative if and only if $\Delta_B|_{\ker\alpha}\neq0$. Thus,
Proposition \ref{com} shows that
\begin{equation}\label{bgc-imp}
\begin{array}{l}
\{\,\bar{C}_0(\Phi_{K,l,\alpha})>0\,\}\,\Leftrightarrow\,\{\,\ker\alpha\neq \{0\}\,\},\\
\{\,\bar{Q}_0(\Phi_{K,l,\alpha})>0\,\}\,\Leftrightarrow\,\{\,\exists\;z_1,z_2\in
\ker\alpha\,\;\textup{such that}\;\, \Delta_B(z_1,z_2)\neq0\,\}.
\end{array}
\end{equation}
In fact, positivity of these capacities means that they are equal to
$+\infty$.\footnote{This follows from the observations in
\cite[Section 4C]{BRC}.}

Since the tensor product of two Gaussian channels
$\Phi_{K_1,l_1,\alpha_1}$ and $\Phi_{K_2,l_2,\alpha_2}$ is a
Gaussian channel $\Phi_{K,l,\alpha}$ with
$\alpha=\alpha_1\oplus\alpha_2$, it is easy to see that equivalence
relations (\ref{bgc-imp}) are valid for the asymptotic zero-error
capacities as well, i.e. for $C_0(\Phi_{K,l,\alpha})$ and
$Q_0(\Phi_{K,l,\alpha})$ instead of $\bar{C}_0(\Phi_{K,l,\alpha})$
and $\bar{Q}_0(\Phi_{K,l,\alpha})$.
\end{example}

\section{Superactivation of one-shot zero-error capacities}

\subsection{The case of zero-error classical capacities}

The superactivation of one-shot zero-error classical capacity means
that
\begin{equation}\label{sa-cc}
    \bar{C}_0(\Phi_1)=\bar{C}_0(\Phi_2)=0,\quad\textrm{but}\quad
    \bar{C}_0(\Phi_1\otimes\Phi_2)>0.
\end{equation}
for some channels $\Phi_1$ and $\Phi_2$. The existence of such
channels was shown independently in \cite{CCH,Duan}. In particular,
in \cite{Duan} an example of two channels\break $\Phi_1\neq\Phi_2$
having input dimension $\dim\H_A=4$ such that (\ref{sa-cc}) holds
was constructed and it was mentioned that this is the minimal input
dimension for which superactivation (\ref{sa-cc}) may take place.
Then by using these two channels and a direct sum construction a
\emph{symmetric} example of superactivation (i.e. (\ref{sa-cc}) with
$\Phi_1=\Phi_2$) with input dimension $\dim\H_A=8$ was obtained
\cite[Theorem 1]{Duan}. In this section we will construct a
symmetric example of superactivation (\ref{sa-cc}) with the minimal
input dimension $\dim\H_A=4$ and the minimal Choi rank $\dim\H_E=3$.

Since a subspace $\mathfrak{L}$ of the algebra $\mathfrak{M}_n$ of
$n\times n$-matrices is a noncommutative graph of a particular
channel if and only if
\begin{equation}\label{L-cond}
\mathfrak{L}\;\,\textup{is
symmetric}\;\,(\mathfrak{L}=\mathfrak{L}^*)\;\,\textup{and contains
the unit matrix}
\end{equation}
(see Lemma 2 in \cite{Duan} and Proposition \ref{cmp} below), Lemma
\ref{trans-l} reduces the problem of finding channels for which
(\ref{sa-cc}) holds to the problem of finding transitive subspaces
$\mathfrak{L}_1$ and $\mathfrak{L}_2$ satisfying (\ref{L-cond}) such
that $\mathfrak{L}_1\otimes\mathfrak{L}_2$ is not transitive. It is
this way that was used in \cite{Duan} to construct the channels
$\Phi_1$ and $\Phi_2$ mentioned above.

It is interesting that the non-preserving of transitivity under
tensor product was known in the theory of operator subspaces: a
transitive subspace $\mathfrak{L}_0\subset\mathfrak{M}_4$ such that
$\mathfrak{L}_0\otimes\mathfrak{L}_0$ is not transitive was
constructed in \cite[Example 3.10]{DMR}. Moreover, the subspace
$\mathfrak{L}_0^{\bot}\doteq\{A\,|\,\Tr
AB=0\;\forall B\in\mathfrak{L}_0\}$ in this example also has the same
property. The above subspaces $\mathfrak{L}_0$
and $\mathfrak{L}_0^{\bot}$ consist respectively of the matrices
$$
\left[\begin{array}{cccc}
a &  b & h &  2g\\
c &  d & f &  e\\
e &  f & a &  b\\
g &  h & c &  d
\end{array}\right],\quad
\left[\begin{array}{cccc}
a &  b & -h &  -g\\
c &  d & -f &  -e\\
e &  f & -a &  -b\\
g/2 &  h & -c & -d
\end{array}\right],\quad a,b,c,d,e,f,g,h\in\mathbb{C}.
$$

This example does not give an example of superactivation of
one-shot zero-error classical capacity, since the subspaces
$\mathfrak{L}_0$ and $\mathfrak{L}_0^{\bot}$ are not symmetric.
Nevertheless, using a similar approach one can construct a
symmetric example.\medskip

\begin{theorem}\label{main+}
\emph{There exists a symmetric transitive subspace
$\mathfrak{L}\subseteq \mathfrak{M}_4$ with $\dim\mathfrak{L}=8$
containing the unit matrix such that
$\mathfrak{L}\otimes\mathfrak{L}$ is not transitive.}
\end{theorem}

\medskip

We will need two lemmas. The first one is similar to  Lemma
2.1 in \cite{DMR}.\smallskip

\begin{lemma}\label{dmr-l} \emph{Let $\,\Phi: \mathfrak{M}_n \to \mathfrak{M}_n$ be
a linear isomorphism with $n^2$ different eigenvalues and such that
all eigenvectors of $\,\Phi^*$ have rank more than or equal to 2.
Then the subspace $$\mathfrak{L} =  \left\{\left[\begin{array}{cc}A&
\Phi(B)\\B&A\end{array}\right]  \;  | \; A, B\in
\mathfrak{M}_n\right\}$$ is transitive.}
\end{lemma}

\begin{proof} Given $z_1, z_2, x, y \in \mathbb C^n$ with $\|x\|^2+\|y\|^2 \neq 0$, we need
to find $A$ and $B$ in $\mathfrak{M}_n$ such that
$$
\left[\begin{array}{cc}A&
\Phi(B)\\B&A\end{array}\right]\left[\begin{array}{c} x
\\y\end{array}\right] = \left[\begin{array}{c} z_1
\\z_2\end{array}\right].
$$

Case 1: $x, y\neq 0, x\neq \lambda y$. Take $B=0$, $A$ such that
$Ax=z_1, Ay = z_2.$

Case 2: $x=0, y\neq 0.$ Take $A$ such that $Ay = z_2$ and $B$ such
that $\Phi(B)y = z_1$ (this is possible, since $\Phi$ is an
isomorphism).

Case 3: $x\neq 0, y=0$. It is similar to the case 2.

Case 4: $x, y\neq 0, x=\lambda y.$ We need to find $A, B$ such that
$$\lambda Ay + \Phi(B)y = z_1, \;\; \lambda By +Ay = z_2.$$
Expressing $Ay$ from the second equation and substituting into the
first one, we get:

\begin{equation}\label{eq-n}
Ay = z_2- \lambda By,
\end{equation}
and  $\lambda z_2 - \lambda^2By + \Phi(B)y = z_1$, whence
$(\Phi(B)-\lambda^2B)y = z_1-\lambda z_2.$ It has a solution if $\,\mathrm{Ran} (\Phi -\lambda^2)\,$ is transitive or,
equivalently, $\mathrm{Ker} (\Phi^* - \overline \lambda^2$) does not
contain a 1-rank operator. If $\lambda^2$ is not an eigenvalue of
$\Phi$ then it holds. If it is an eigenvalue, then this kernel is a
1-dimensional subspace generated by a matrix of rank $\ge 2$, so it
again holds. And now one finds $A$ from (\ref{eq-n}).\end{proof}

\begin{lemma}\label{tpt-l} \emph{Let $\,\mathfrak{L}$ be a subspace of $\,\mathfrak{M}_n$. The
subspace $\,\mathfrak{L}\otimes \mathfrak{L}$ is transitive if and
only if the subspace
$\mathfrak{L}A\mathfrak{L}^{\top}\doteq\{\sum_iX_iAY_i^{\top}\,|\,X_i,Y_i\in\mathfrak{L}\}$
coincides with $\,\mathfrak{M}_n$ for each $A\in \mathfrak{M}_n$}.
\end{lemma}

\begin{proof}
We may identify $\mathbb C^n\otimes \mathbb C^n$ with
$\mathfrak{M}_n$  by the linear isomorphism $U:
x\otimes y \mapsto x\cdot y^{\top}$ (we assume that $x,y$ are columns).

There exists a linear isomorphism $\Lambda: \B(\mathbb C^n\otimes
\mathbb C^n) \to \B(\mathfrak{M}_n)$ given by $\Lambda(T\otimes S) =
L_TR_{S^{\top}}$ (left multiplication by $T$ and right
multiplication by $S^{\top}$), which agrees with $U$ in the sense
that
$$U[T\otimes S]z = \Lambda(T\otimes
S)Uz\quad\forall z\in\mathbb C^n\otimes
\mathbb C^n.
$$
This implies the assertion of the lemma.
\end{proof}

\bigskip

{\it Proof of Theorem \ref{main+}.}
Let $$C_1 = \left[\begin{array}{cc} 0 & \mathrm{i} \\1 & 0
\end{array}\right], \; C_2 = \left[\begin{array}{cr} 0 & -\mathrm{i} \\1 & 0
\end{array}\right], \; C_3 = \left[\begin{array}{cr} 1 & 0 \\0 & 1
\end{array}\right], \; C_4 = \left[\begin{array}{cr} 1 & 0 \\0 & -1
\end{array}\right].$$ These matrices form an orthogonal basis in
$\mathfrak{M}_2$.  Let $\lambda_1=\mathrm{i}, \lambda_2 = -\mathrm{i}, \lambda_3 = 1,
\lambda_4 = -1.$ We define an unitary map $\Phi: \mathfrak{M}_2 \to
\mathfrak{M}_2$ by $\Phi(C_i) = \lambda_i C_i.$

Let $\mathfrak{L}= \left\{\left[\begin{array}{cc}A&
\Phi(B)\\B&A\end{array}\right]  \; | \; A, B\in
\mathfrak{M}_2\right\}$ be a subspace of $\mathfrak{M}_4$. Since
$\Phi\left(\left[\begin{array}{cc}a&
b\\c&d\end{array}\right]\right)=\left[\begin{array}{cr}d&
-c\\b&a\end{array}\right]$, the subspace $\mathfrak{L}$ consists of
the matrices
$$
\left[\begin{array}{cccc}
a &  b & h &  -g\\
c &  d & f &  e\\
e &  f & a &  b\\
g &  h & c &  d
\end{array}\right],\quad a,b,c,d,e,f,g,h\in\mathbb{C}.
$$

It is
clear that
$\dim\mathfrak{L}=\dim\mathfrak{M}_2+\dim\mathfrak{M}_2=8$ and that the subspace $\mathfrak{L}$ is symmetric.
\smallskip
Transitivity of $\mathfrak{L}$  follows from Lemma \ref{dmr-l}.
\smallskip

To prove that $\mathfrak{L}\otimes \mathfrak{L}$ is not transitive
it suffices, by Lemma \ref{tpt-l},  to show that
$\mathfrak{L}\left[\begin{array}{cc}1&
0\\0&-1\end{array}\right]\mathfrak{L}^{\top} \neq \mathfrak{M}_4.$
We have

\begin{equation}\label{1}
\!\!\!\!\!\begin{array}{l} \mathfrak{L}\left[\begin{array}{cc}1&
0\\0&-1\end{array}\right]\mathfrak{L}^{\top}=
\\\\
\displaystyle \left\{\sum_i \left[\begin{array}{cc}A^1_i&
\Phi(B^1_i)\\B^1_i&A^1_i\end{array}\right]\!
\left[\begin{array}{cc}1&
0\\0&-1\end{array}\right]\!\left[\begin{array}{cc}A^2_i& \!\!\!B_i^{2\top}\\
\Phi(B^2_i)^{\top}& \!\!\!A^2_i\end{array}\right] |\, A^{1,2}_i,
B^{1,2}_i \in \mathfrak{M}_2\right\} = \\\\ \displaystyle
\left\{\sum_i \left[\begin{array}{cc}A^1_i&
-\Phi(B^1_i)\\B^1_i&-A^1_i\end{array}\right]\!
\left[\begin{array}{cc}A^2_i&
B_i^{2\top}\\\Phi(B^2_i)^{\top}&A^2_i\end{array}\right] |\,
A^{1,2}_i, B^{1,2}_i
 \in \mathfrak{M}_2\right\} = \\\\ \displaystyle
\left\{\left[\begin{array}{cc} \!\sum_i (A^1_iA^2_i -
\Phi(B^1_i)\Phi(B^2_i)^{\top})\!\! & \ldots \\\ldots & \!\!\sum_i
(B^1_iB_i^{2\top} - A^1_iA^2_i)\!\!\end{array}\right] | \,A^{1,2}_i,
B^{1,2}_i \in \mathfrak{M}_2\right\}\!\!\!\!\!
\end{array}
\end{equation}

Let $B^1, B^2\in \mathfrak{M}_2$. We can write them as $B^1 = \sum_i
t_iC_i, B^2 = \sum_i s_iC_i$. Since $\Tr C_iC_j^{\top}\neq 0$ only
in the cases: a)~$i=1, j=2$, b)~$i=2, j=1$, c)~$i=j=3$, d)~$i=j =
4$, we obtain

\begin{equation}\label{2}
\!\!\!\!\begin{array}{l} \displaystyle \Tr (B^1B^{2\top} -
\Phi(B^1)\Phi(B^2)^{\top}) = \\\\ \displaystyle \Tr
\left(\sum_{i,j} t_is_jC_iC_j^{\top} - \sum_{i,j} \lambda_it_i\lambda_js_jC_iC_j^{\top}\right) = \\\\
\displaystyle \sum_{i,j} \Tr
(1-\lambda_i\lambda_j)t_is_jC_iC_j^{\top} = \\\\ \displaystyle
\Tr(1-\lambda_1\lambda_2)t_1s_2C_1C_2^{\top} +
\Tr(1-\lambda_2\lambda_1)t_2s_1C_2C_1^{\top} + \\\\ \displaystyle
+\Tr (1-\lambda_3^2)t_3s_3C_3C_3^{\top} + \Tr
(1-\lambda_4^2)t_4s_4C_4C_4^{\top} = 0.
\end{array}\!\!\!
\end{equation}
It follows from (\ref{1}) and (\ref{2}) that for any $T\in
\mathfrak{L}\left[\begin{array}{cc}1&
0\\0&-1\end{array}\right]\mathfrak{L}^{\top}$ we have
$$
\Tr
(T_{11}+T_{22})=0.
$$
Thus $\mathfrak{L}\left[\begin{array}{cc}1&
0\\0&-1\end{array}\right]\mathfrak{L}^{\top} \neq \mathfrak{M}_4.$
\qed

\medskip

To derive from Theorem \ref{main+} an example of superactivation of
one-shot  zero-error classical capacity with smallest possible
dimension we need the following observation (which is a strengthened
version of Lemma 2 in \cite{Duan}).\smallskip

\begin{property}\label{cmp} \emph{Let $\,\mathfrak{L}$ be a subspace of
$\,\mathfrak{M}_n$, $n\geq2$, and $m$ the minimal natural number
such that $\,\dim\mathfrak{L}\leq m^2$. The following statements are
equivalent:}
\begin{enumerate}[(i)]
    \item \emph{$\mathfrak{L}$ is symmetric ($\mathfrak{L}^*=\mathfrak{L}$) and contains the unit
    matrix;}
    \item \emph{there exists an entanglement-breaking channel $\,\Psi:\mathfrak{M}_n\rightarrow\mathfrak{M}_m$
    such that $\mathfrak{L}=\Psi^*(\mathfrak{M}_m)$ ($\,\Psi^*:\mathfrak{M}_m\rightarrow\mathfrak{M}_n$ is a dual map to the channel $\Psi$).}
    \item \emph{there exists a pseudo-diagonal \footnote{A channel $\Phi:\T(\H_A)\rightarrow\T(\H_B)$ is called pseudo-diagonal
    if it has the representation
$$
\Phi(\rho)=\sum_{i,j}c_{ij}\langle
\psi_i|\rho|\psi_j\rangle|i\rangle\langle j|,\quad\rho\in\T(\H_A),
$$
where $\{c_{ij}\}$ is a Gram matrix of a collection of unit vectors,
$\{|\psi_i\rangle\}$ is a collection of vectors in $\H_A$ such that
$\;\sum_i |\psi_i\rangle\langle \psi_i|=I_{\H_A}\,$  and
$\{|i\rangle\}$ is an orthonormal basis in $\H_B$ \cite{R}.} channel
$\,\Phi:\mathfrak{M}_n\rightarrow\mathfrak{M}_{nm}$ with the Choi
rank $\,m$
    such that $\,\mathfrak{L}=\mathcal{G}(\Phi)$ (the noncommutative graph of
    $\,\Phi$).}
\end{enumerate}
\end{property}

\begin{proof} $\mathrm{(ii)\Rightarrow (i)}$ is obvious.\smallskip

$\mathrm{(i)\Rightarrow (ii)}$.  We will show first that there is a
basis $\{A_i\}_{i=1}^d$ of $\mathfrak{L}$ with all $A_i$'s being
positive such that $\sum_{i=1}^d A_i=I_n$ (the unit matrix in
$\mathfrak{M}_n$). It is sufficient to show that such a basis exists
in the real space
$\mathfrak{L}_{sa}=\{A\in\mathfrak{L}\,|\,A=A^*\}$, since it will
also be a basis for $\mathfrak{L}$ over $\mathbb C$ (by symmetricity
of $\mathfrak{L}$). Since any ball generates the whole space, we can
find a basis $I_n, \tilde A_2, \ldots, \tilde A_n$ with all $\tilde
A_i$ belonging to a ball in $\mathfrak{L}_{sa} $ with centrum $I_n$
and of radius, say,  $1/2$. Since for any $A=A^*\in \mathfrak{M}_n$,
$\|I_n-A\| < 1$ implies that $A\ge 0$, we conclude that $\tilde A_i
\ge 0$. Now let $M$ be a sufficiently large number such that
$I_n-\sum_{i=2}^n \tilde{A}_i/M \ge 0$. Let $A_1 = I_n-\sum_{i=2}^n
\tilde{A}_i/M.$ It is easy to see that $A_1, \tilde A_2, \ldots,
\tilde A_n$ form a basis and $$I_n = A_1 + \sum_{i=2}^n \tilde
A_i/M.$$ Now take $A_i = \tilde A_i/M, i = 2,\ldots, n$.
\smallskip

Let $\{B_i\}_{i=1}^d$,
$d=\dim\mathfrak{L}$, be a set of positive linearly independent
matrices in $\mathfrak{M}_m$ with unit trace. Consider the unital completely positive map
$$
\mathfrak{M}_m\ni X\mapsto\Psi^*(X) = \sum_{i=1}^d[\Tr B_i
X]A_i\in\mathfrak{M}_n
$$
Apparently $\mathrm{Ran} \Psi^* \subseteq \mathfrak{L}$. To see that it is
exactly $\mathfrak{L}$, we will show that each $A_i$ is in the
range. For that we just take any $X\in \mathfrak{M}_m$ such that
$\Tr B_j X=0$ for all $j\neq i$ and $\Tr B_i X\neq 0,$ which exists
since $B_i$'s are linearly independent.

Since the map $\Psi^*$ has the Kraus representation consisting of
1-rank operators, the predual map
$\Psi:\mathfrak{M}_n\rightarrow\mathfrak{M}_m$ is an
entanglement-breaking quantum channel.

$\mathrm{(ii)\Leftrightarrow (iii)}$  It suffices to note that a
pseudo-diagonal channel is complementary to an entanglement-breaking
channel and vice versa \cite{R}.
\end{proof}

The proof of Proposition \ref{cmp} can be used to obtain an explicit
formula for a channel $\Phi$ with given noncommutative
graph.\smallskip

\begin{corollary}\label{cmp-c}
\emph{Let $\,\mathfrak{L}$ be a subspace of $\,\mathfrak{M}_n$,
$n\geq2$, satisfying (\ref{L-cond}) and $m$ the minimal natural
number such that $\,d=\dim\mathfrak{L}\leq m^2$. There is a
pseudo-diagonal channel $\,\Phi$ with $\,\dim\H_{A}=n$,
$\,\dim\H_{E}=m$ and $\,\dim\H_{B}\leq mn$ such that
$\mathcal{G}(\Phi)=\mathfrak{L}$ represented as follows
\begin{equation}\label{ch-rep}
\mathfrak{M}_n\ni\rho\mapsto\Phi(\rho)=\sum_{i,j=1}^d
c_{ij}A_i^{1/2}\rho A_j^{1/2}\otimes|i\rangle\langle
j|\in\mathfrak{M}_n\otimes\mathfrak{M}_d,
\end{equation}
where $\{A_i\}_{i=1}^d$ is a basis of $\,\mathfrak{L}$ such that
$\,\sum_{i=1}^d A_i=I_n$ and $A_i\geq0$ for all $i$, $\{c_{ij}\}$ is
the Gram matrix of a set $\,\{|\psi_i\rangle\}_{i=1}^d$ of unit
vectors in $\mathbb{C}^m$ such that the set
$\,\{|\psi_i\rangle\langle\psi_i|\}_{i=1}^d$ is linearly independent
and $\{|i\rangle\}$ is the canonical basis in $\mathbb{C}^d$.}
\end{corollary}\medskip

By representation (\ref{ch-rep}) the channel $\Phi$ maps a state
$\rho\in\mathfrak{M}_n$ into the $d\times d$ matrix
$\left[c_{ij}A_i^{1/2}\rho A_j^{1/2}\right]$ with entries in
$\mathfrak{M}_n$. Its formal output dimension $nd$ may be greater
than $mn$, but the real output dimension is $\leq mn$ (since $\Phi$
is complementary to a channel from $\mathfrak{M}_n$ into
$\mathfrak{M}_m$, see the proof). If $d>m$ this means that all the
states $\Phi(\rho)$ in (\ref{ch-rep}) are supported by a proper
subspace $\H_0\subset\mathbb{C}^n\otimes\mathbb{C}^d$ such that
$\dim\H_0\leq mn$.

\begin{proof} The proof of Proposition \ref{cmp} shows that a channel $\Phi$ with the
stated properties can be constructed as the complementary channel to
the channel
$$
\Psi(\rho) = \sum_{i=1}^d[\Tr A_i \rho]B_i,
$$
where $\{A_i\}\subset\mathfrak{M}_n$ is a basis of $\mathfrak{L}$
determined in that proof and $\{B_i\}\subset\mathfrak{M}_m$ is any
linearly independent set of positive matrices with unit trace. We
may assume that $B_i=|\psi_i\rangle\langle\psi_i|$ for all
$i=\overline{1,d}$, where $\{|\psi_i\rangle\}_{i=1}^d$ is a set of
unit vectors in $\mathbb{C}^m$ such that the set
$\,\{|\psi_i\rangle\langle\psi_i|\}_{i=1}^d$ is linearly
independent. Consider the linear operator
$$
V:|\varphi\rangle\mapsto\sum_{i=1}^dA_i^{1/2}|\varphi\rangle\otimes|i\rangle\otimes|\psi_i\rangle
$$
from $\mathbb{C}^n$ into
$\mathbb{C}^n\otimes\mathbb{C}^d\otimes\mathbb{C}^m$, where
$\{|i\rangle\}$ is the canonical basis in $\mathbb{C}^d$.\smallskip

Since $\sum_{i=1}^d A_i=I_n$ and $\|\psi_i\|=1$ for all $i$, $V$ is
an isometry. It is easy to see that
$$
\Tr_{\mathbb{C}^n\otimes\mathbb{C}^d}V|\varphi\rangle\langle\varphi|V^*=\sum_{i=1}^d[\Tr
A_i
|\varphi\rangle\langle\varphi|]|\psi_i\rangle\langle\psi_i|,\quad
\varphi\in\mathbb{C}^n.
$$
So, $\Psi(\rho)=\Tr_{\mathbb{C}^n\otimes\mathbb{C}^d}V\rho V^*$ and
hence
$$
\Phi(\rho)=\widehat{\Psi}(\rho)=\Tr_{\mathbb{C}^m}V\rho
V^*=\sum_{i,j=1}^d \langle\psi_j|\psi_i\rangle A_i^{1/2}\rho
A_j^{1/2}\otimes|i\rangle\langle j|,\quad \rho\in\mathfrak{M}_n.
$$
\end{proof}

Using the subspace $\mathfrak{L}$
from Theorem \ref{main+} and applying Proposition \ref{cmp}
we obtain  the following corollary.\smallskip

\begin{corollary}\label{main-c+}
\emph{There is a  pseudo-diagonal channel
$\,\Phi:\T(\H_{A})\rightarrow\T(\H_{B})$ with $\,\dim\H_{A}=4$,
$\,\dim\H_{E}=3$ and  $\,\dim\H_{B}\leq12$,  for which the
following symmetric form of superactivation of one-shot zero-error
classical capacity holds:}
\begin{equation}\label{sa-cc+}
    \bar{C}_0(\Phi)=0,\quad\textit{but}\quad
    \bar{C}_0(\Phi\otimes\Phi)>0.
\end{equation}
\end{corollary}\smallskip
By finding a basis $\{A_i\}_{i=1}^8$ of $\,\mathfrak{L}$ such that
$\,\sum_{i=1}^8 A_i=I_4$ and $A_i\geq0$ for all $i$ and applying
Corollary \ref{cmp-c} one can obtain an explicit expression for a
channel $\Phi$ having the properties stated in Corollary
\ref{main-c+}.

In \cite[Theorem 1]{Duan} the same statement  was established with
$\dim\H_{A}=8$ and it was mentioned that (\ref{sa-cc+}) does not
hold for any channel $\Phi$ with $\dim\H_{A}<4$. So, Corollary
\ref{main-c+} gives a symmetric example of superactivation of
one-shot zero-error classical capacity with \emph{minimal input
dimension $\,\dim\H_{A}$ and minimal Choi rank $\,\dim\H_{E}$}.
Minimality of $\,\dim\H_{E}=3\,$ follows from the fact that any
transitive subspace of $\mathfrak{M}_4$ has dimension $\geq7$
\cite{DMR}.

\subsection{The extreme form of superactivation}

According to the notations in \cite{C&S}, the extreme form of
superactivation of one-shot zero-error capacity means the existence
of two channels $\Phi_1$ and $\Phi_2$ such that
\begin{equation}\label{ef-sa}
    \bar{C}_0(\Phi_1)=\bar{C}_0(\Phi_2)=0,\quad\textrm{but}\quad
    \bar{Q}_0(\Phi_1\otimes\Phi_2)>0.
\end{equation}
Since $\bar{Q}_0$ is less than or equal to $\bar{C}_0$, the channels
$\Phi_1$ and $\Phi_2$ demonstrate superactivation of both classical
and quantum one-shot zero-error capacities simultaneously, i.e.
(\ref{sa-cc}) and
\begin{equation}\label{sa-qc}
    \bar{Q}_0(\Phi_1)=\bar{Q}_0(\Phi_2)=0,\quad\textrm{but}\quad
    \bar{Q}_0(\Phi_1\otimes\Phi_2)>0.
\end{equation}

In \cite{C&S} a very sophisticated method is used to show the
existence of two channels $\Phi_1$ and $\Phi_2$ of sufficiently high
dimensions ($\dim\H_A=48, \dim\H_E=1140, \dim\H_B=54720$) for which
the extreme form of superactivation of asymptotic zero-error
capacity holds (which means validity of (\ref{ef-sa}) with
$\bar{C}_0$ and $\bar{Q}_0$ replaced by $C_0$ and $Q_0$).

This result directly implies the existence of two channels $\Phi_1$
and $\Phi_2$ for which (\ref{ef-sa}) holds, but it neither gives an
explicit form of these channels, nor says anything about their
minimal dimensions.

We want to fill this gap and present a low-dimensional example of
such channels expressed in terms of their noncommutative
graphs.\smallskip

By Lemmas \ref{trans-l} and \ref{trans-l+} (with Proposition
\ref{cmp}) the problem of finding channels for which (\ref{ef-sa})
holds is reduced to the problem of finding transitive subspaces
$\mathfrak{L}_1\subset\mathfrak{M}_{n_1}$ and
$\mathfrak{L}_2\subset\mathfrak{M}_{n_2}$ satisfying (\ref{L-cond})
such that
\begin{equation}\label{i-r}
\langle \psi|A|\varphi\rangle=0\quad\textup{and}\quad \langle
\varphi|A|\varphi\rangle=\langle \psi|A|\psi\rangle\quad\forall
A\in\mathfrak{L}_1\otimes\mathfrak{L}_2
\end{equation}
for some unit vectors $\varphi$ and $\psi$ in
$\mathbb{C}^{n_1}\otimes\mathbb{C}^{n_2}$.
\smallskip

Let $A\mapsto\widehat{A}$ be the linear isomorphism of
$\mathfrak{M}_4$ corresponding to the Shur multiplication by the
matrix
$$
T = [t_{ij}]=\left[\begin{array}{rrrr}
1 &  1 & -\mathrm{i} &  -\mathrm{i}\\
1 &  1 & -\mathrm{i} &  -\mathrm{i}\\
+\mathrm{i} &  +\mathrm{i} & 1 &  1\\
+\mathrm{i} &  +\mathrm{i} & 1 &  1
\end{array}\right],
$$
i.e. $\{\hat{a}_{ij}\}=\{a_{ij}t_{ij}\}$, and $\mathfrak{L}_0$ the
subspace of $\mathfrak{M}_4$ constructed in Example 3.10 in
\cite{DMR} ($\mathfrak{L}_0$ and $\mathfrak{L}_0^{\bot}$ are
described in Subsection 3.1). Consider the subspaces
$$
\mathfrak{L}_1 = \left\{M_1=\left[\begin{array}{cc} A_1 &
 B_1\\C_1& \widehat{A}_1\end{array}\right],\;\; A_1\in \mathfrak{M},\, B_1, C_1^*\in \mathfrak{L}_0^{\bot}\right\}
$$
and
$$
\mathfrak{L}_2 = \left\{M_2=\left[\begin{array}{cc} \widehat{A}_2 &
 B_2\\ C_2 & A_2 \end{array}\right],\;\; A_2\in \mathfrak{M}, \,B_2, C_2^*\in \mathfrak{L}_0^{\bot}\right\},
$$
of $\mathfrak{M}_8$, where $\mathfrak{M}$ is a subspace of
$\mathfrak{M}_4$ having the properties stated in Lemma \ref{main-l}
below. Since $\dim\mathfrak{L}_0^{\bot}=8$,
$\dim\mathfrak{L}_1=\dim\mathfrak{L}_2=8+8+7=23$.
\medskip

Since $[\widehat{A}\,]^*=\widehat{[A^*]}$ and $\widehat{I}_4=I_4$,
the subspaces $\mathfrak{L}_1$ and $\mathfrak{L}_2$ are symmetric
and contain the unit matrix $I_8$. It is easy to see that they are
transitive (since
$\mathfrak{L}_0^{\bot},[\mathfrak{L}_0^{\bot}]^*,\mathfrak{M}$ and
$\widehat{\mathfrak{M}}\doteq\{\widehat{A}\,|\,A\in\mathfrak{M}\}$
are transitive subspaces of $\mathfrak{M}_4$).\smallskip

\begin{theorem}\label{main}
\emph{There exist unit vectors $\varphi$ and $\psi$ in
$\mathbb{C}^8\otimes\mathbb{C}^8$ such that $(\ref{i-r})$ holds for
the above transitive subspaces $\,\mathfrak{L}_1$ and
$\,\mathfrak{L}_2$ of $\,\mathfrak{M}_8$}.
\end{theorem}

\begin{proof} We have to show the existence of two orthogonal unit
vectors $\varphi, \psi$ in
$[\mathbb{C}^4\oplus\mathbb{C}^4]\otimes[\mathbb{C}^4\oplus\mathbb{C}^4]$
such that
\begin{equation}\label{one}
    \langle\psi | M_1\otimes M_2 | \varphi\rangle=0\quad \forall
    M_1\in\mathfrak{L}_1, M_2\in\mathfrak{L}_2
\end{equation}
and
\begin{equation}\label{two}
    \langle\psi | M_1\otimes M_2 |\psi\rangle=\langle\varphi | M_1\otimes M_2 | \varphi\rangle\quad \forall
    M_1\in\mathfrak{L}_1, M_2\in\mathfrak{L}_2.
\end{equation}

Let $|u\rangle=\sum_{i=1}^4 |x_i\rangle\otimes |y_i\rangle$ and
$|v\rangle=\sum_{i=1}^4 s_i |x_i\rangle\otimes |y_i\rangle$ be the
vectors in $\mathbb{C}^4\otimes\mathbb{C}^4$, where
$|x_i\rangle=|e_i\rangle$, $|y_i\rangle=|e_{5-i}\rangle$
($\{|e_i\rangle\}$ is the canonical basis in $\mathbb{C}^4$) and
$s_1=s_2=1, s_3=s_4=-1$. It is shown in \cite{DMR} that
$|u\rangle\langle v|\in[\mathfrak{L}_0^{\bot}\otimes
\mathfrak{L}_0^{\bot}]^{\bot}$, which means that
\begin{equation}\label{one+}
0=\langle v | B_1\otimes B_2 | u\rangle=\sum_{i,j=1}^4 s_i\langle
x_i\otimes y_i  | B_1\otimes B_2 |\, x_j\otimes y_j\rangle\quad
\forall B_1,B_2\in \mathfrak{L}_0^{\bot}.
\end{equation}

Let $|\varphi\rangle=\frac{1}{2}\sum_{i=1}^4 |0, x_i\rangle\otimes
|0, y_i\rangle$ and $|\psi\rangle=\frac{1}{2}\sum_{i=1}^4 s_i
|x_i,0\rangle\otimes |y_i,0\rangle$. Then we have
\begin{equation}\label{int}
M_1\otimes M_2|\varphi\rangle=\frac{1}{2}\sum_{i=1}^4 |B_1 x_i,
\widehat{A}_1x_i\rangle\otimes|B_2 y_i, A_2y_i\rangle
\end{equation}
and hence
$$
\begin{array}{c}
\displaystyle\langle\psi| M_1\otimes M_2 |
\varphi\rangle=\frac{1}{4}\sum_{i,j=1}^4 s_i \langle x_i,0|\otimes
\langle y_i,0 |\cdot| B_1 x_j,
\widehat{A}_1x_j\rangle\otimes| B_2 y_j, A_2y_j\rangle\\
\displaystyle=\frac{1}{4}\sum_{i,j=1}^4 s_i \langle x_i |B_1
|x_j\rangle\langle y_i | B_2 | y_j\rangle=0,\quad
\end{array}
$$
where the last equality follows from (\ref{one+}). Thus (\ref{one})
is valid. It follows from (\ref{int}) that
\begin{equation}\label{two+}
\begin{array}{c}
\displaystyle\langle\varphi| M_1\otimes M_2 |
\varphi\rangle=\frac{1}{4}\sum_{i,j=1}^4 \langle 0,x_i|\otimes
\langle 0, y_i|\cdot|B_1 x_j,
\widehat{A}_1x_j\rangle\otimes|B_2 y_j, A_2y_j\rangle\\
\displaystyle=\frac{1}{4}\sum_{i,j=1}^4 \langle x_i | \widehat{A}_1
| x_j\rangle\langle y_i |A_2 | y_j\rangle=\frac{1}{4}\sum_{i,j=1}^4
t_{ij}a^1_{ij}a^2_{k(i)k(j)},\quad k(i)=5-i,
\end{array}
\end{equation}
where $a^n_{ij}$ are elements of the matrix $A_n, n=1,2$. Since
\begin{equation*}
M_1\otimes M_2 |\psi\rangle=\frac{1}{2}\sum_{i=1}^4 s_i|A_1 x_i,
C_1x_i\rangle\otimes|\widehat{A}_2 y_i, C_2y_i\rangle,
\end{equation*}
we have
\begin{equation*}
\begin{array}{c}
\displaystyle\langle\psi| M_1\otimes M_2 |
\psi\rangle=\frac{1}{4}\sum_{i,j=1}^4 s_i s_j\langle x_i,0|\otimes
\langle y_i,0| \cdot |A_1 x_j,
C_1x_j\rangle\otimes|\widehat{A}_2 y_j, C_2y_j\rangle\\
\displaystyle=\frac{1}{4}\sum_{i,j=1}^4 s_is_j\langle
x_i|A_1|x_j\rangle\langle y_i
|\widehat{A}_2|y_j\rangle=\frac{1}{4}\sum_{i,j=1}^4
s_is_jt_{k(i)k(j)}a^1_{ij}a^2_{k(i)k(j)},\,\; k(i)=5-i.
\end{array}
\end{equation*}
The right hand side of this equality coincides with the right hand
side of (\ref{two+}), since it is easy to verify that
$t_{ij}=s_is_jt_{k(i)k(j)}$. Hence (\ref{two}) is valid.
\end{proof}

\begin{lemma}\label{main-l} \emph{There exists a  transitive
subspace $\,\mathfrak{M}$ of $\,\mathfrak{M}_4$ with
$\dim\mathfrak{M}=7$ satisfying (\ref{L-cond}) such that the
subspace
$\,\widehat{\mathfrak{M}}\doteq\{\widehat{A}\,|\,A\in\mathfrak{M}\}$,
where $A\mapsto\widehat{A}$ is the above-defined isomorphism, is
transitive (and satisfies (\ref{L-cond})).}
\end{lemma}

\begin{proof} The proof below is essentially based on the
arguments from the proof of Theorem 1.2 in \cite{DMR}.

Consider the subspace $\mathfrak{N}\subset\mathfrak{M}_4$ consisting
of the matrices
$$
\left[\begin{array}{cccc}
a+b+c &  f+g & i &  0\\
d+e &  -a & 2f+g   &  i\\
h &  2d+e & -b &  3f+g \\
0 &  h & 3d+e &  -c
\end{array}\right],
$$
where $\,a,b,c,d,e,f,g,h,i\,$ are complex numbers.\smallskip

This subspace does not contain 1-rank matrices. Indeed, a non-zero
matrix $N$ of $\mathfrak{N}$ is non-zero on some diagonal. Consider
the square submatrix containing the shortest non-zero diagonal of
$N$ as its main diagonal. This submatrix is triangular, and hence
its rank is not less than the rank of its diagonal, which is at
least $2$. Hence $\rank N\geq2$.

Let $\mathfrak{M}=\mathfrak{N}^{\perp}\doteq\{A\,|\,\Tr
AB=0\;\forall B\in\mathfrak{N}\}$. Since the subspace $\mathfrak{N}$
is symmetric and consists of traceless matrices of rank $\neq 1$,
$\mathfrak{M}$ is a symmetric transitive subspace containing the
unit matrix. Since $\dim\mathfrak{N}=9$, $\dim\mathfrak{M}=16-9=7$.

To complete the proof it suffices to show that the subspace
$\widehat{\mathfrak{M}}$ is transitive. This can be done by checking
that $\Tr\widehat{A}\widehat{B}=\Tr AB$ for any
$A,B\in\mathfrak{M}_4$, which implies
$\widehat{\mathfrak{M}}=[\widehat{\mathfrak{N}}]^{\perp}$, and by
verifying that the subspace $\widehat{\mathfrak{N}}$ does not
contain 1-rank matrices (in the same way as for $\mathfrak{N}$).
\end{proof}

Theorem \ref{main} and Proposition \ref{cmp} immediately imply the
following result.\smallskip

\begin{corollary}\label{main-c}
\emph{There exists a pair of pseudo-diagonal channels
$\;\Phi_i:\T(\H_{A_i})\rightarrow\T(\H_{B_i})$ with
$\dim\H_{A_i}=8$, $\dim\H_{E_i}=5$ and  $\dim\H_{B_i}\leq40$,
$i=1,2$, for which extreme superactivation (\ref{ef-sa}) holds.}
\end{corollary}\smallskip

By using Corollary \ref{cmp-c} one can obtain explicit expressions
for channels $\Phi_1$ and $\Phi_2$ having the properties stated in
Corollary \ref{main-c}.

Since the subspaces $\mathfrak{L}_1$ and $\mathfrak{L}_2$ are not
unitary equivalent, the above example of extreme superactivation is
essentially nonsymmetric: $\Phi_1\neq\Phi_2$. But they can be used
to construct a symmetric example by applying the direct sum
construction (see the proof of Theorem 1 in \cite{Duan}).
\smallskip

\begin{corollary}\label{main-c++}
\emph{There exists a quantum channel
$\,\Phi:\T(\H_{A})\rightarrow\T(\H_{B})$ with $\,\dim\H_{A}=16$,
$\dim\H_{E}=10$ and  $\,\dim\H_{B}\leq40$, for which the
following symmetric form of the extreme superactivation holds:}
\begin{equation*}
    \bar{C}_0(\Phi)=0,\quad\textit{but}\quad
    \bar{Q}_0(\Phi\otimes\Phi)>0.
\end{equation*}
\end{corollary}\smallskip
This means that the channel $\Phi$ has vanishing one-shot classical
zero-error capacity but positive two-shot quantum zero-error
capacity.

\section{On channels which cannot be superactivated}

J.Park and S.Lee showed in \cite{P&L} that superactivation of
one-shot zero-error classical capacity (\ref{sa-cc}) does not hold
if either $\Phi_1$ or $\Phi_2$ is a qubit channel.\footnote{In fact,
one can prove that superactivation of one-shot zero-error classical
capacity (\ref{sa-cc}) does not hold if either $\Phi_1$ or $\Phi_2$
has input dimension $\leq 3$ \cite{D-pc}.} Now we will show how to
substantially extend this observation by using some results from
\cite{DMR} and \cite{M&S}, in particular, the following lemma (which
is a reformulation of Corollary 6.13 in \cite{DMR}).
\smallskip

\begin{lemma}\label{b-c} \emph{Let $\,\mathfrak{L}_1$ be a transitive subspace of $\,\B(\H_1)$ which is contained
in the weak-operator-topology closed linear span of its 1-rank
elements. Then the spatial tensor product $\,\mathfrak{L}_1\otimes
\mathfrak{L}_2$ is a transitive subspace of $\,\B(\H_1\otimes\H_2)$
for any transitive subspace $\,\mathfrak{L}_2$ of $\,\B(\H_2)$.}
\end{lemma}\smallskip

This observation is a strengthened infinite-dimensional version of
the well known fact that the tensor product of any two unextendible
product base is an unextendible product base \cite{Shor&Co}.
\smallskip

\begin{property}\label{non-sa} \emph{Superactivation (\ref{sa-cc}) of one-shot zero-error classical capacity does not
hold for two channels
$\,\Phi_i:\T(\H_{A_i})\rightarrow\T(\H_{B_i})$, $i=1,2$  if the
channel $\,\Phi_1$ satisfies one of the following conditions (in
which $\,\mathcal{G}(\Phi_1)\doteq\widehat{\Phi}^*_1(\B(\H_{E_1}))$
is the non-commutative graph of $\,\Phi_1$):}
\begin{enumerate}[A)]
    \item $\dim\mathcal{G}(\Phi_1)\geq[\dim\H_{A_1}]^2-1\;$ ($\,\dim\H_{A_1}<+\infty$);
    \item \emph{$\dim\H_{A_1}=2$, in particular, $\Phi_1$ is a qubit channel};
    \item \emph{$\mathcal{G}(\Phi_1)$ is an algebra;}
    \item \emph{$\Phi_1$ is a Bosonic Gaussian channel (described in Example 1);}
    \item \emph{$\Phi_1$ is a finite-dimensional entanglement-breaking channel;}
    \item \emph{$\Phi_1$ is an entanglement-breaking channel having Kraus representation (\ref{Kraus-rep}) such that
    $\,\rank V_k=1$ for all $\,k$,}\footnote{This means that $\Phi\otimes\mathrm{Id}_{\K}(\omega)$
    is a countably-decomposable separable state in $\S(\H_B\otimes\K)$ for any state $\omega\in\S(\H_A\otimes\K)$, see Remark \ref{non-sa-r} below.}
    \end{enumerate}
\emph{and the channel $\,\Phi_2$ is arbitrary.}
\end{property}

\begin{proof} A) If $\,\mathcal{G}(\Phi_1)=\B(\H_{A_1})$  then this assertion follows from assertion C.
If $\,\dim\mathcal{G}(\Phi_1)= [\dim\H_{A_1}]^2-1\,$ then
$\,\dim\ker \widehat{\Phi}_1=1$. If the one-shot zero-error
classical capacity of the channel $\Phi_1$ is zero then, by Lemma
\ref{rank1}, the minimal rank of all nonzero operators in
$\ker\widehat{\Phi}_1 $ is not less than $2$. By \cite[Theorem
1.1]{M&S} this implies that the subspace $\ker\widehat{\Phi}_1$ is
reflexive, which means that
$\mathcal{G}(\Phi_1)=[\ker\widehat{\Phi}_1]^{\bot}$ is spanned by
its one rank elements \cite[Claim 3.1]{M&S}.

If $\Phi_2$ is an arbitrary channel with zero one-shot zero-error
classical capacity then $\mathcal{G}(\Phi_2)$ is a  transitive
subspace (by Lemma \ref{trans-l}). Lemma \ref{b-c} shows that
$\mathcal{G}(\Phi_1\otimes\Phi_2)=\mathcal{G}(\Phi_1)\otimes\mathcal{G}(\Phi_2)$
is a transitive subspace and hence the one-shot zero-error classical
capacity of the channel $\Phi_1\otimes\Phi_2$ is zero (by Lemma
\ref{trans-l}).\smallskip

B) If $\dim\H_{A_1}=2$ and  $\bar{C}_0(\Phi_1)=0$  then, by Lemma
\ref{rank1}, the all nonzero operators in $\ker\widehat{\Phi}_1$
have rank $=2$, i.e they are invertible. This implies that
$\dim\ker\widehat{\Phi}_1\leq1$. Indeed, if  $T, S$ are invertible
operators in $\ker\widehat{\Phi}_1$ and $\lambda$ is an eigenvalue
of the operator
 $TS^{-1}$ then $$T-\lambda S = (TS^{-1}-\lambda)S$$ is a
 non-invertible operator in $\ker\widehat{\Phi}_1$ and hence  $T=\lambda S$. So, this assertion
follows from the previous one.\smallskip

C) If $\mathcal{G}(\Phi_1)$ is an algebra and
$\,\bar{C}_0(\Phi_1)=0\,$ then Proposition \ref{com} and the basic
results of the von Neumann algebras theory (cf.\cite{K&R}) imply
that $\mathcal{G}(\Phi_1)$ is dense in $\B(\H_{A_1})$ in the
weak-operator topology.  Hence to prove that
$\,\bar{C}_0(\Phi_1\otimes\Phi_2)=0\,$ for any channel $\Phi_2$ with
$\,\bar{C}_0(\Phi_2)=0\,$ it suffices, by Lemma \ref{trans-l}, to
show transitivity of the subspace
$\,\B(\H_{A_1})\otimes\mathfrak{L}\,$ for any transitive subspace
$\mathfrak{L}$ of $\B(\H_{A_2})$.

This assertion is obvious if $\,n=\dim\H_{A_1}<+\infty$, since in
this case the subspace $\B(\H_{A_1})\otimes\mathfrak{L}$ can be
identified with the subspace of all $\,n\times n\,$  matrices with
entries in $\mathfrak{L}$ (considered as operators in
$\bigoplus_{k=1}^n \H_k$, where $\H_k$ is a copy of $\H_{A_2}$ for
all $k$).

Assume that $\dim\H_{A_1}=+\infty$ and there is a vector
$|\varphi\rangle=\sum_{i=1}^{+\infty}c_i|e_i\otimes f_i\rangle$ in
$\H_{A_1}\otimes\H_{A_2}$ (where $c_1\neq0$, $\{|e_i\rangle\}$ and
$\{|f_i\rangle\}$ are orthonormal base in $\H_{A_1}$ and in
$\H_{A_2}$) such that all the vectors $C|\varphi\rangle$,
$C\in\B(\H_{A_1})\otimes\mathfrak{L}$, belong to a proper subspace
$\K$ of $\H_{A_1}\otimes\H_{A_2}$.  Let $\H_n$ be the subspace of
$\H_{A_1}$ spanned by the vectors $|e_1\rangle,\ldots,|e_n\rangle$
and $\,|\varphi_n\rangle=\sum_{i=1}^{n}c_i|e_i\otimes f_i\rangle$.
By the above observation the set
$\,\{\,C|\varphi_n\rangle\,|\,C\in\B(\H_n)\otimes\mathfrak{L}\,\}\,$
is dense in $\H_n\otimes\H_{A_2}$. But it is easy to see that
$$
C|\varphi_n\rangle=C|\varphi\rangle
$$
for any $C\in\B(\H_n)\otimes\mathfrak{L}$. Since
$\B(\H_n)\otimes\mathfrak{L}\subseteq\B(\H_{A_1})\otimes\mathfrak{L}$
this implies $\H_n\otimes\H_{A_2}\subseteq\K\,$ for any $n$, that is
a contradiction. \smallskip

D) This assertion follows from the previous one, since the
noncommutative graph of a Bosonic Gaussian channel is an algebra
(see Example \ref{gauss}).\smallskip

E) If $\Phi_1$ is a finite-dimensional entanglement-breaking channel
then it has Kraus representation (\ref{Kraus-rep}) such that
$\,\rank V_k=1$ for all $\,k$ \cite{e-b-ch}. So, this assertion
follows from assertion F.

F) In this case the noncommutative graph
$\,\mathcal{G}(\Phi_1)\doteq\widehat{\Phi}^*_1(\B(\H_{E_1}))$ is
spanned by the 1-rank operators $V_k^*V_l$ (this follows from
expression (\ref{Kraus-rep-c})).  So, this assertion follows from
Lemmas \ref{trans-l} and  \ref{b-c}.
 \end{proof}

Proposition \ref{non-sa} directly implies the following two
observations.\smallskip

\begin{corollary}\label{non-sa-c-1} \emph{If a quantum channel $\,\Phi$ satisfies one of  conditions A-F
from Proposition \ref{non-sa} then $\,C_0(\Phi)=0$ if and only if
$\,\bar{C}_0(\Phi)=0$.}
\end{corollary}
\medskip

\begin{corollary}\label{non-sa-c-2}
\emph{Superactivation of asymptotic classical zero-error capacity
(property (\ref{sa-cc}) with $\bar{C}_0$ replaced by $C_0$) does not
hold for channels $\,\Phi_1$ and $\,\Phi_2$, if $\,\Phi_1$ satisfies
one of conditions A-F from Proposition \ref{non-sa} and $\,\Phi_2$
is arbitrary.}\end{corollary}
\medskip

\begin{remark}\label{non-sa-r}
The question about validity of the assertions of Proposition
\ref{non-sa} and Corollaries \ref{non-sa-c-1}-\ref{non-sa-c-2} for
arbitrary infinite-dimensional entanglement-breaking channel
$\Phi_1$ remains open, since the existence of countably
nondecomposable separable states in an infinite-dimensional
bipartite quantum system implies the existence of
entanglement-breaking channels which don't have Kraus representation
(\ref{Kraus-rep}) with 1-rank operators $V_k$ \cite{HSW}.
\end{remark}\medskip

\begin{property}\label{non-sa+} \emph{Superactivation (\ref{sa-qc}) of one-shot zero-error quantum capacity does not
hold for two channels
$\,\Phi_i:\T(\H_{A_i})\rightarrow\T(\H_{B_i})$, $i=1,2$  if one of
the following conditions holds (in which
$\,\mathcal{G}(\Phi_i)\doteq\widehat{\Phi}^*_i(\B(\H_{E_i}))$ is the
noncommutative graph of $\,\Phi_i$):}
\begin{enumerate}[A)]
    \item \emph{$\mathcal{G}(\Phi_1)$ contains a maximal commutative $*$-subalgebra of
    $\,\mathfrak{M}_{n_1}$,  where $n_1=\dim\H_{A_1}<+\infty$, and $\,\Phi_2$ is an arbitrary channel;}
    \item \emph{$\dim\H_{A_1}=2$ (in particular, when $\Phi_1$ is a qubit channel) and $\,\Phi_2$ is an arbitrary channel};
    \item \emph{$\mathcal{G}(\Phi_1)$ and $\,\mathcal{G}(\Phi_2)$ are algebras;}
    \item \emph{$\Phi_1$ and $\,\Phi_2$ are Bosonic Gaussian channels (described in Example 1).}
\end{enumerate}
\end{property}

\begin{proof} A) Since a  maximal commutative $*$-subalgebra of
    $\,\mathfrak{M}_{n_1}$ consists of all matrices which are diagonal with respect to some orthonormal basis, the noncommutative graph $\mathcal{G}(\Phi_1\otimes\Phi_2)$
contains the subspace of all block-diagonal  matrices of the form
$\mathrm{diag}(a_1A, ..., a_{n_1}A)$, where $a_1, \ldots, a_{n_1}\in\mathbb{C}$
and $A \in\mathcal{G}(\Phi_2)$. So, the assumption
$\bar{Q}_0(\Phi_1\otimes\Phi_2)>0$ implies, by Lemma \ref{trans-l+},
the existence of unit vectors $|\varphi\rangle=(x_1, ... , x_{n_1})$ and
$|\psi\rangle=(y_1, ... , y_{n_1}),$ where $x_k,y_k\in\H_{A_2}$, such
that
$$
\sum_{k=1}^{n_1} a_k\langle y_k|A|x_k\rangle=0\quad \textrm{and}
\quad \sum_{k=1}^{n_1} a_k\langle x_k|A|x_k\rangle=\sum_{k=1}^{n_1}
a_k \langle y_k|A|y_k\rangle
$$
for all $a_1, \ldots, a_{n_1}\in\mathbb{C}$ and all $A
\in\mathcal{G}(\Phi_2)$. It follows  that
\begin{equation}\label{tmp-rel}
\langle y_k|A|x_k\rangle=0\quad \textrm{and} \quad \langle
x_k|A|x_k\rangle=\langle y_k|A|y_k\rangle
\end{equation}
for all $k$ and all $A \in\mathcal{G}(\Phi_2)$. Since
$\mathcal{G}(\Phi_2)$ contains the identity operator,
(\ref{tmp-rel}) shows  that $\|x_{k}\|=\|y_{k}\|$ for all $k$ and
hence there exists $k_0$ such that $\|x_{k_0}\|=\|y_{k_0}\|\neq0$.
Thus, (\ref{tmp-rel}) with $k=k_0$ implies, by Lemma \ref{trans-l+},
that $\bar{Q}_0(\Phi_2)>0$.\medskip

B) follows from assertion A, since the noncommutative graph
$\mathcal{G}(\Phi_1)$ of any non-reversible channel $\Phi_1$ with
$\dim\H_{A_1}=2$ contains a maximal commutative $*$-subalgebra of
$\,\mathfrak{M}_{2}$. Indeed, since $\mathcal{G}(\Phi_1)$ contains
an operator $T\neq\lambda I_2$, it contains a self-adjoint operator
$T'\neq\lambda I_2$ which is diagonal in a particular basis. The
operators $T'$ and $I_2$ generate a maximal commutative
$*$-subalgebra of $\,\mathfrak{M}_{2}$ contained in
$\mathcal{G}(\Phi_1)$.\medskip

C) follows from Proposition \ref{com}, since
$$
[\mathcal{G}(\Phi_1\otimes\Phi_2)]'
=[\mathcal{G}(\Phi_1)]'\,\bar{\otimes}\,[\mathcal{G}(\Phi_2)]',
$$
where $\bar{\otimes}$ denotes a tensor product of von Neumann
algebras \cite[Ch.10]{K&R}.

D) This assertion follows from the previous one, since the
noncommutative graph of a Bosonic Gaussian channel is an algebra
(see Example \ref{gauss}).
\end{proof}

Proposition \ref{non-sa+} and its proof imply the following two
observations.\smallskip

\begin{corollary}\label{non-sa+c-1} \emph{If a quantum channel $\,\Phi$ satisfies one of  conditions
A-D from Proposition \ref{non-sa+} then $\,Q_0(\Phi)=0$ if and only
if $\,\bar{Q}_0(\Phi)=0$.}
\end{corollary}
\medskip

\begin{corollary}\label{non-sa+c-2}
\emph{Superactivation of asymptotic quantum zero-error capacity
(property (\ref{sa-qc}) with $\bar{Q}_0$ replaced by $Q_0$) does not
hold for channels $\,\Phi_1$ and $\,\Phi_2$ satisfying one of
conditions A-D from Proposition \ref{non-sa+}.}
\end{corollary}

\section{Relations to reversibility  properties of a channel}

\subsection{Reversibility of a single channel and one-shot zero-error
capacities}

Reversibility (sufficiency) of a quantum channel
$\Phi:\mathfrak{T}(\mathcal{H}_A)\rightarrow\mathfrak{T}(\mathcal{H}_B)$
with respect to a family $\S$ of states in
$\mathfrak{S}(\mathcal{H}_A)$ means the existence of a quantum
channel
$\Psi:\mathfrak{T}(\mathcal{H}_B)\rightarrow\mathfrak{T}(\mathcal{H}_A)$
such that $\Psi(\Phi(\rho))=\rho$ for all $\rho\in\S$
\cite{J&P,J-rev}.

The notion of reversibility of a channel naturally arises in
analysis of different general questions of quantum information
theory, in particular, of conditions for preserving entropic
characteristics of quantum states under the action of a channel. In
particular, it follows from Petz's theorem that the Holevo
quantity\footnote{The Holevo quantity provides an upper bound for
accessible classical information which can be obtained by applying a
quantum measurement \cite{H-SCI,N&Ch}.} of an ensemble
$\{\pi_i,\rho_i\}$ of quantum states is preserved under the action
of a quantum channel $\Phi$, i.e.
$$
\chi(\{\pi_i,\Phi(\rho_i)\})=\chi(\{\pi_i,\rho_i\}),
$$
if and only if the channel $\Phi$ is reversible with respect to the
family $\{\rho_i\}$ \cite{J&P}.\smallskip

A general criterion for reversibility of a quantum channel (in the
 von Neumann algebras theory settings) is obtained in
\cite{J&P}. Several conditions for reversibility expressed in terms
of a complementary channel are derived from this criterion in
\cite{BRC}, where a complete characterization of reversibility with
respect to families of pure states is given. The case of families of
pure states is of special interest in quantum information theory,
since many capacity-like characteristics of a quantum channel can be
determined as extremal values of functionals depending on ensembles
of pure states \cite{H-SCI, N&Ch}.\smallskip

To describe reversibility properties of a channel $\Phi$ the
\emph{reversibility index}
$$\ri(\Phi)=[\,\ri_1(\Phi),\ri_2(\Phi)\,]$$ is introduced in
\cite{BRC}, in which the components $\ri_1(\Phi)$ and $\ri_2(\Phi)$
take the values $0,1,2$. The first component $\ri_1(\Phi)$
characterizes reversibility of the channel $\Phi$ with respect to
(w.r.t.) complete\footnote{A family
$\{|\varphi_{\lambda}\rangle\langle\varphi_{\lambda}|\}_{\lambda\in\Lambda}$
of pure states in $\S(\H)$ is called complete if  the linear hull of
the family $\{|\varphi_{\lambda}\rangle\}_{\lambda\in\Lambda}$ is
dense in $\H$.} families $\S$ of pure states as follows
\begin{description}
    \item [$\ri_1(\Phi)=0\,$] if $\,\Phi\,$  is not reversible w.r.t. any
    complete family $\S$;
    \item [$\ri_1(\Phi)=1\,$] if $\,\Phi\,$  is reversible w.r.t. a complete orthogonal family $\S$ but it is not reversible w.r.t. any
    complete nonorthogonal family $\S$;
    \item [$\ri_1(\Phi)=2\,$] if $\,\Phi\,$ is reversible w.r.t. a complete nonorthogonal family
    $\S$.
\end{description}

The second component $\ri_2(\Phi)$ characterizes reversibility of
the channel  $\Phi$ with respect to noncomplete families of pure
states and is defined similarly to $\ri_1(\Phi)$ with the term
"complete" replaced by "noncomplete".\smallskip

So that $\,\ri(\Phi)=01\,$ means that the channel $\Phi$  is not
reversible with respect to any family of pure states which is either
complete or nonorthogonal, but it is reversible with respect to some
noncomplete orthogonal family.\medskip

A channel $\Phi$  with given $\,\ri(\Phi)\,$ can be characterized by
properties of the set $\ker\widehat{\Phi}$ \cite[Corollary 2]{BRC}.
This characterization and  Lemmas \ref{rank1},\ref{trans-l+} show
that
$$
\ri_2(\Phi)=0 \;\Leftrightarrow\; \bar{C}_0(\Phi)=0,\qquad
\ri_2(\Phi)=2 \;\Leftrightarrow\; \bar{Q}_0(\Phi)>0,
$$
while $\ri_2(\Phi)=1$ means that $\bar{C}_0(\Phi)>0$ but
$\bar{Q}_0(\Phi)=0$.

\subsection{On reversibility of a tensor product channel}

Let $\,\Phi:\T(\H_A)\rightarrow\T(\H_B)$ and
$\,\Psi:\T(\H_C)\rightarrow\T(\H_D)$ be arbitrary quantum channels.
It is easy to see that reversibility of the channels $\Phi$ and
$\Psi$ with respect to particular families $\S_{\Phi}$ and
$\S_{\Psi}$ imply reversibility of the channel $\Phi\otimes\Psi$
with respect to the family
$\S_{\Phi}\otimes\S_{\Psi}=\{\rho\otimes\sigma\,|\,\rho\in\S_{\Phi},
\sigma\in\S_{\Phi}\}$. It follows that
\begin{equation}\label{ineq-1}
    \ri_1(\Phi\otimes\Psi)\geq\min\{\ri_1(\Phi),\ri_1(\Psi)\}
\end{equation}
and
\begin{equation}\label{ineq-2}
    \ri_2(\Phi\otimes\Psi)\geq\max\{\ri_2(\Phi),\ri_2(\Psi)\}.
\end{equation}

An interesting question concerns the possibility of a strict
inequality in (\ref{ineq-1}) and in (\ref{ineq-2}). This question is
nontrivial, since the channel $\Phi\otimes\Psi$ may be reversible
with respect to families consisting of \emph{entangled} pure states
in $\S(\H_A\otimes\H_C)$ (and the corresponding reversing channel
may not be of the tensor product form).
\smallskip

As to inequality (\ref{ineq-1}) this question has a simple
solution.\smallskip
\begin{property}\label{eq-1}
\emph{An equality holds in (\ref{ineq-1}) for any channels $\,\Phi$
and $\,\Psi$.}
\end{property}

\begin{proof} This follows from Corollary 2 in \cite{BRC}, since it
is easy to show that $\widehat{\Phi}\otimes\widehat{\Psi}$ is a
discrete c-q channel if and \emph{only if} $\widehat{\Phi}$ and
$\widehat{\Psi}$ are discrete c-q channels.\footnote{A channel
$\Phi:\T(\H_A)\rightarrow\T(\H_B)$ is called discrete
classical-quantum (discrete c-q)  if it has the representation
$\Phi(\rho)=\sum_{i=1}^{\dim\H_A}\langle i|\rho|i\rangle\sigma_i,$
where $\{|i\rangle\}$ is an orthonormal basis in $\H_A$ and
$\{\sigma_i\}$ is a collection of states in $\S(\H_B)$
\cite{H-SCI}.} \end{proof}

By the remark at the end of Section 5.1 the validity of a strict
inequality in (\ref{ineq-2}) means a particular form of
superactivation of one-shot zero-error capacities. For example, the
superactivation of one-shot zero-error classical capacity is
equivalent to the existence of two channels $\Phi_1$ and $\Phi_2$
such that
$$
\ri_2(\Phi_1)=\ri_1(\Phi_2)=0,\quad \textup{but}\quad
\ri_2(\Phi_1\otimes\Phi_2)=1,
$$
while the extreme form of superactivation means the existence of two
channels $\Phi_1$ and $\Phi_2$ such that
$$
\ri_2(\Phi_1)=\ri_2(\Phi_2)=0,\quad \textup{but}\quad
\ri_2(\Phi_1\otimes\Phi_2)=2.
$$
These effects can be also called \emph{superactivation of
reversibility} of a channel.\smallskip

So, we see that reversibility of a channel with respect to
\emph{noncomplete} families of pure states can be superactivated by
tensor products in contrast to reversibility with respect to
\emph{complete} families of pure states (this follows from
Proposition \ref{eq-1}).\smallskip

Proposition \ref{non-sa} shows that
$$
\ri_2(\Phi_1)=\ri_2(\Phi_2)=0\quad \Rightarrow\quad
\ri_2(\Phi_1\otimes\Phi_2)=0
$$
for any channel $\Phi_1$ satisfying one of the conditions of this
proposition and arbitrary channel $\Phi_2$.

Proposition \ref{non-sa+} shows that
$$
\max\{\ri_2(\Phi_1),\ri_2(\Phi_2)\}<2\quad \Rightarrow\quad
\ri_2(\Phi_1\otimes\Phi_2)<2
$$
for any channels $\Phi_1$ and $\Phi_2$ satisfying one of the
conditions of this proposition.\bigskip

We are grateful to A.S.~Holevo and to the participants of his
seminar "Quantum probability, statistic, information" (the Steklov
Mathematical Institute) for useful discussion. We are also grateful
to R.~Duan for comments concerning minimal dimension of channels
demonstrating the superactivation of one-shot zero-error classical
capacity. We would like to thank V.S.~Shulman and P.B.M.~Sorensen
for helping with some particular questions. We are grateful to Dan
Stahlke for pointing a mistake in the previous version of the paper and
to the unknown referee for valuable suggestions.
\bigskip

The work of the first-named author is partially supported the
fundamental research programs of the Russian Academy of Sciences and
by the RFBR grant 13-01-00295a. Research of the second-named author
is  funded by the Polish National Science Centre grant under the
contract number DEC-2012/06/A/ST1/00256.

\end{document}